\documentclass{llncs}
\usepackage{makeidx}  
\usepackage{amssymb,latexsym,amsfonts,amsmath}
\usepackage{graphicx}
\usepackage{eso-pic}
\usepackage{epsfig}
\usepackage{epstopdf}
\usepackage{comment}
\usepackage{stmaryrd}
\usepackage{tikz}
\DeclareSymbolFont{symbolsC}{U}{pxsyc}{m}{n}
\SetSymbolFont{symbolsC}{bold}{U}{pxsyc}{bx}{n}
\DeclareFontSubstitution{U}{pxsyc}{m}{n}
\DeclareMathSymbol{\medcirc}{\mathbin}{symbolsC}{7}
\usepackage{booktabs}
\usepackage{subfigure}
\usepackage{listings}
\usepackage{color}
\usepackage{wrapfig}
\usepackage{algorithm}
\usepackage[noend]{algpseudocode}
\algnewcommand\algorithmicinput{\textbf{Input:}}
\algnewcommand\Input{\item[\algorithmicinput]}
\algnewcommand\algorithmicoutput{\textbf{Output:}}
\algnewcommand\Output{\item[\algorithmicoutput]}
\usepackage[normalem]{ulem}

\newcommand{\R}{{\mathbb{R}}}
\newcommand{\N}{{\mathbb{N}}}

\newtheorem{assumption}{Assumption}%

\begin{document}
	\title{Temporal Logic Verification of Stochastic Systems Using Barrier Certificates\thanks{This work was supported in part by the German Research Foundation (DFG) through the grant ZA 873/1-1 and the TUM International Graduate School of Science and Engineering (IGSSE).}}
	\author{
		Pushpak Jagtap\inst{1}
		\and
		Sadegh Soudjani\inst{2}
		\and
		Majid Zamani\inst{1}
	}
		\institute{Technical University of Munich, Germany\\
		\email{\{pushpak.jagtap,zamani\}@tum.de} 
		\and
		 Newcastle University, United Kingdom\\
		 \email{Sadegh.Soudjani@ncl.ac.uk} }
	\maketitle
	\begin{abstract}
	This paper presents a methodology for temporal logic verification of discrete-time stochastic systems. Our goal is to find a lower bound on the probability that a complex temporal property is satisfied by finite traces of the system. Desired temporal properties of the system are expressed using a fragment of linear temporal logic, called \emph{safe LTL over finite traces}. 
We propose to use barrier certificates for computations of such lower bounds, which is computationally much more efficient than the existing discretization-based approaches. The new approach is discretization-free and does not suffer from the curse of dimensionality caused by discretizing state sets.
The proposed approach relies on decomposing the negation of the specification into a union of sequential reachabilities and then using barrier certificates to compute upper bounds for these reachability probabilities.
We demonstrate the effectiveness of the proposed approach on case studies with linear and polynomial dynamics.
	\end{abstract}

	\section{Introduction}\label{sect_intro}
	Verification of dynamical systems against complex specifications has gained significant attention in last few decades \cite{baier2008principles,tabuada2009verification}.
	The verification task is challenging for continuous-space dynamical systems under uncertainties and is hard to be performed exactly.
	%
	There have been several results in the literature utilizing approximate finite models (a.k.a. abstractions) for verification of stochastic dynamical systems.
	Examples include results on verification of discrete-time stochastic hybrid systems against probabilistic invariance \cite{esmaeil2013adaptive,SAprecise14} and linear temporal logic specifications  \cite{abate2011quantitative,tkachev2013formula} using Markov chain abstractions.
	Verification of discrete-time stochastic switched systems against probabilistic computational tree logic formulae is discussed in \cite{7029024} using interval Markov chains as abstract models.
	However, these abstraction techniques are based on state set discretization and face the issue of discrete state explosion. This scalability issue is only partly mitigated in \cite{SAM15,LSZ17} based on compositional abstraction of stochastic systems.
	
	On the other hand, a discretization-free approach, based on \emph{barrier certificates}, has been used for verifying stochastic systems against simple temporal properties such as safety and reachability.
	Employing barrier certificates for safety verification of stochastic systems is initially proposed in \cite{4287147}. Similar results are reported in \cite{8263999} for switched diffusion processes and piecewise-deterministic Markov processes. The results in \cite{huang2017probabilistic} propose a probabilistic barrier certificate to compute bounds on the probability that a stochastic hybrid system reaches unsafe region. However, in order to provide infinite time horizon guarantees, all of these results require an assumption that the barrier certificates exhibit \emph{supermartingale} property which in turns presuppose stochastic stability and vanishing noise at the equilibrium point of the system.
	
	In this work, we consider the problem of verifying discrete-time stochastic systems against complex specifications over finite time horizons \emph{without} requiring any assumption on the stability of the system. This is achieved by relaxing supermartingale condition to \emph{$c$-martingale} as also utilized in \cite{steinhardt2012finite}. Correspondingly, instead of infinite-horizon specifications, we consider finite-horizon temporal specifications, which are more practical in the real life applications including motion planning problems \cite{6942758,7158879,6224963}.
	In spirit, this work extends the idea of combining automata representation of the specification and barrier certificates, which is proposed in \cite{7364197} for non-stochastic dynamics, in order to verify stochastic systems against specifications expressed as a fragment of LTL formulae, namely, safe LTL on finite traces. Our work also has the same flavour as \cite{DM14}, but in a completely different setting, in combining barrier certificates to guarantee satisfaction of temporal specifications.
	
	To the best of our knowledge, this paper is the first one to use barrier certificates for algorithmic verification of stochastic systems against a wide class of temporal properties. Our main contribution is to provide a systematic approach for computing lower bounds on the probability that the discrete-time stochastic system satisfies given safe LTL specification over a finite time horizon. This is achieved by first decomposing specification into a sequence of simpler verification tasks based on the structure of the automaton associated with the negation of the specification. Next, we use barrier certificates for computing probability bounds for simpler verification tasks which are further combined to get a (potentially conservative) lower bound on the probability of satisfying the original specification. The effectiveness of the proposed approach is demonstrated using several case studies with linear and polynomial dynamics.

\section{Preliminaries}\label{II}
\subsection{Notations}
We denote the set of nonnegative integers by $\N_0 := \{0, 1, 2, \ldots\}$ and the set of positive integers by $\N := \{1, 2, 3, \ldots \}$. The symbols $ \R$, $\R^+,$ and $\R_0^+ $ denote the set of real, positive, and nonnegative real numbers, respectively. We use $ \mathbb{R}^{n\times m} $ to denote the space of real matrices with $ n $ rows and $ m $ columns. \\
We consider a probability space $ (\Omega,\mathcal{F}_\Omega,\mathbb{P}_\Omega) $ where $\Omega$ is the sample space, $\mathcal{F}_\Omega$ is a sigma-algebra on $\Omega$ comprising the subset of $\Omega$ as events, and $\mathbb{P}_\Omega$ is a probability measure that assigns probabilities to events. We assume that random variables introduced in this article are measurable functions of the form $X:(\Omega,\mathcal{F}_\Omega)\rightarrow(S_X,\mathcal{F}_X)$ as $Prob\{A\}=\mathbb{P}_\Omega\{X^{-1}(A)\}$ for any $A\in\mathcal{F}_X$. We often directly discuss the probability measure on $(S_X,\mathcal{F}_X)$ without explicitly mentioning the underlying probability space and the function $X$ itself.


\subsection{Discrete-time stochastic systems}
In this work, we consider discrete-time stochastic systems given by a tuple $S=(X,V_{\textsf w},w,f)$, where $X$ and $V_{\textsf w}$ are Borel spaces representing state and uncertainty spaces of the system. We denote by $(X,\mathcal{B}(X))$ the measurable space with $\mathcal{B}(X)$ being the Borel sigma-algebra on the state space. Notation $w$ denotes a sequence of independent and identically distributed (i.i.d.) random variables on the set $V_{\textsf w}$ as $w:=\{w(k):\Omega\rightarrow V_{\textsf w}, \ k\in\mathbb N_0\}$.  
The map $f : X \times V_{\textsf w} \rightarrow X$ is a measurable function characterizing the state evolution of the system. For a given initial state $x(0)\in X$, the state evolution can be written as
\begin{equation}
 x(k+1)=f(x(k),w(k)), \ \ \ k\in\mathbb{N}_0.
\label{sys}
\end{equation}
We denote the \emph{solution process} generated over $N$ time steps by $\textbf{x}_N=x(0),$ $x(1), \ldots$, $x(N-1)$.
The sequence $w$ together with the measurable function $f$ induce a unique probability measure on the sequences $\textbf{x}_N$. 

We are interested in computing a lower bound on the probability that system $S=(X,V_{\textsf w},w,f)$ satisfies a specification expressed as a temporal logic property.  We provide syntax and semantics of the class of specifications dealt with in this paper in the next subsection.

\subsection{Linear temporal logic over finite traces}
\label{sec:LTL}
In this subsection, we introduce linear temporal logic over finite traces, referred to as LTL$_F$ \cite{de2013linear}. LTL$_F$ uses the same syntax of LTL over infinite traces given in \cite{baier2008principles}. The LTL$_F$ formulas over a set $ \Pi $ of atomic propositions are obtained as follows:
\begin{equation}
 \varphi ::=  \text{ true} \mid p \mid \neg \varphi \mid \varphi_1 \wedge \varphi_2 \mid \varphi_1 \vee\varphi_2\mid\medcirc \varphi  \mid  \lozenge\varphi \mid \square\varphi \mid \varphi_1\mathcal{U}\varphi_2, \nonumber
 \end{equation}
 where $p \in \Pi$, $\medcirc $ is the next operator, $\lozenge$ is eventually, $\square$ is always, and $\mathcal{U}$ is until. The semantics of LTL$_F$ is given in terms of \textit{finite traces}, \textit{i.e.}, finite words $\sigma$, denoting a finite non-empty sequence of consecutive steps over $\Pi$. We use $|\sigma |$ to represent the length of $\sigma$ and $\sigma_i$ as a propositional interpretation at $i$th position in the trace, where $0\leq i < |\sigma |$. Given a finite trace $\sigma$ and an LTL$_F$ formula $\varphi$, we inductively define when an LTL$_F$ formula $\varphi$ is true at the $i$th step $(0\leq i <|\sigma |)$, denoted by $\sigma,i\models\varphi$, as follows:
 \begin{itemize}
 \item $\sigma,i\models \text{true}$;
  \item $\sigma,i\models p$, for $p\in\Pi$ iff $p\in\sigma_i$;
   \item $\sigma,i\models \neg\varphi$ iff $\sigma,i\not\models\varphi$;
    \item $\sigma,i\models \varphi_1\wedge\varphi_2$ iff $\sigma,i\models\varphi_1$ and $\sigma,i\models\varphi_2$;
      \item $\sigma,i\models \varphi_1\vee\varphi_2$ iff $\sigma,i\models\varphi_1$ or $\sigma,i\models\varphi_2$;  
     \item $\sigma,i\models \medcirc \varphi$ iff $i<|\sigma |-1$ and $\sigma,i+1\models\varphi$; 
     \item $\sigma,i\models  \lozenge\varphi$ iff for some $j$ such that $i\leq j< |\sigma |$, we have $\sigma,j\models\varphi$; 
     \item $\sigma,i\models  \square\varphi$ iff for all $j$ such that $i\leq j<|\sigma |$, we have $\sigma,j\models\varphi$; 
     \item $\sigma,i\models  \varphi_1\mathcal{U}\varphi_2$ iff for some $j$ such that $i\leq j< |\sigma |$, we have $\sigma,j\models\varphi_2$, and for all $k$ s.t. $i\leq k<j$, we have $\sigma,k\models\varphi_1$. 
 \end{itemize}
 The formula $\varphi$ is true on $\sigma$, denoted by $\sigma\models\varphi$, if and only if $\sigma,0\models\varphi$. We denote the language of such finite traces associated with LTL$_F$ formula $\varphi$ by $\mathcal{L}(\varphi)$.  Notice that in this case we also have the usual boolean equivalences such as $\varphi_1\vee\varphi_2\equiv \neg(\neg\varphi_1\wedge\neg\varphi_2)$, 
 $\varphi_1\implies\varphi_2\equiv \neg\varphi_1\vee\varphi_2$, $\lozenge\varphi \equiv \text{true } \mathcal{U} \varphi$, and $\square\varphi\equiv\neg\lozenge\neg\varphi$.
 
 In this paper, we consider only safety properties \cite{kupferman1999model}. Hence, we use a subset of LTL$_F$ called safe LTL$_F$ as introduced in \cite{6942758} and defined as follows.
 \begin{definition}\label{safe_LTL_f}
 An LTL$_F$ formula is called a safe LTL$_F$ formula if it can be represented in positive normal form, \textit{i.e.}, negations only occur adjacent to atomic propositions, using the temporal operators next $(\medcirc)$ and always $(\square)$.
 \end{definition}
 
Next, we define deterministic finite automata which later serve as equivalent representations of LTL$_F$ formulae.
\begin{definition}\label{DFA}
A deterministic finite automaton $($DFA$)$ is a tuple $\mathcal{A}=(Q,Q_0,\Sigma,$ $\delta,F)$, where $Q$ is a finite set of states, $Q_0\subseteq Q$ is a set of initial states, $\Sigma$ is a finite set $($a.k.a. alphabet$)$, $\delta: Q\times\Sigma\rightarrow Q$ is a transition function, and $F\subseteq Q$ is a set of accepting states.
\end{definition}
We use notation $q\overset{\sigma}{\longrightarrow} q'$ to denote transition relation $(q,\sigma,q')\in\delta$.
A finite word $\sigma=(\sigma_0,\sigma_1,\ldots,\sigma_{n-1})\in \Sigma^n$ is accepted by a DFA $\mathcal{A}$ if there exists a finite state run $q=(q_0,q_1,\ldots,q_{n})\in Q^{n+1}$ such that $q_0\in Q_0$, $q_i \overset{\sigma_i}{\longrightarrow} q_{i+1}$ for all $0\leq i< n$ and $q_{n}\in F$. The accepted language of $\mathcal{A}$, denoted by $\mathcal{L}(\mathcal{A})$, is the set of all words accepted by $\mathcal{A}$.\\
According to \cite{de2015synthesis}, every LTL$_F$ formula $\varphi$
can be translated to a DFA $\mathcal{A}_\varphi$
that accepts the same language as $\varphi$, \textit{i.e.}, $\mathcal{L}(\varphi)=\mathcal{L}(A_\varphi)$. Such $\mathcal{A}_\varphi$ can be constructed explicitly or symbolically using existing tools, such as SPOT \cite{duret2016spot} and MONA \cite{henriksen1995mona}.

\begin{remark}
For a given LTL$_F$ formula $\varphi$ over atomic propositions $\Pi$, the associated DFA $\mathcal A_\varphi$ is usually constructed over the alphabet $\Sigma = 2^\Pi$.
Solution process of a system $S$ is also connected to the set of words by a labeling function $L$ from the state space to the alphabet $\Sigma$. Without loss of generality, we work with the set of atomic propositions directly as the alphabet rather than its power set.
\end{remark}

\medskip
\noindent
\textbf {Property satisfaction by the solution process.}
For a given discrete-time stochastic system $S = (X,V_{\textsf w}, w,f)$ with dynamics \eqref{sys}, finite-time solution processes $\textbf{x}_N$ are connected to LTL$_F$ formulae with the help of a measurable labeling function $L: X \rightarrow \Pi$, where $\Pi$ is the set of atomic propositions.
%
\begin{definition}\label{sys_trace}
For a stochastic system $S = (X,V_{\textsf w},w,f)$ and labeling function $L: X \rightarrow \Pi$, a finite sequence $\sigma_{\textbf{x}_N}=(\sigma_0,\sigma_1,\ldots,\sigma_{N-1})\in\Pi^N$
is a finite trace of the solution process $\textbf{x}_N=x(0)$, $x(1)$,$\ldots$, $x(N-1)$ of $S$ if we have $\sigma_k=L(x(k))$ for all $k \in\{0,1,\ldots,N-1\}$.
 \end{definition}
 Next, we define the probability that the discrete-time stochastic system $S$ satisfies safe LTL$_F$ formula $\varphi$ over traces of length $|\sigma |=N$.
 \begin{definition}
 Let $Trace_{ N}(S)$ be the set of all finite traces of solution processes of $S$
 with length $|\sigma_{\textbf{x}_N}|= N$ and $\varphi$ be a safe LTL$_F$ formula over $\Pi$. Then $\mathbb{P}\{Trace_{ N}(S) \models \varphi\}$ is the probability that $\varphi$ is satisfied by discrete-time stochastic system $S$ over a finite time horizon $[0,N)\subset\N_0$.
 \end{definition}
 
\begin{remark}
The set of atomic propositions $\Pi=\{p_0,p_1,\ldots,p_M\}$ and the labeling function $L: X \rightarrow \Pi$ provide a measurable partition of the state space $X = \cup_{i=1}^M X_i$ as  $X_i:=L^{-1}(p_i)$. Without loss of generality, we assumed that $X_i\neq \emptyset$ for any $i$.
\end{remark}


\subsection{Problem formulation}
\begin{problem}\label{prob}
Given a system $S = (X,V_{\textsf w}, w, f)$ with dynamics \eqref{sys}, a safe LTL$_F$ specification $\varphi$ of length $N$ over a set $\Pi=\{p_0,p_1,\ldots,p_M\}$ of atomic propositions,
and a labeling function $L: X \rightarrow \Pi$,
 compute a lower bound on the probability that the traces of solution process of $S$ of length $N$ satisfies $\varphi$, \textit{i.e.}, a quantity $\vartheta$ such that $\mathbb{P}\{Trace_{ N}(S) \models \varphi\}\ge \vartheta$.
\end{problem}
Note that $\vartheta = 0$ is a trivial lower bound, but we are looking at computation of lower bounds that are as tight as possible.
For finding a solution to Problem \ref{prob}, we first compute an upper bound on the probability $\mathbb{P}\{Trace_{N}(S)\models\neg\varphi\}$.
This is done by constructing a DFA $\mathcal{A}_{\neg \varphi}=(Q,Q_0,\Pi,\delta,F)$ that accepts all finite words over $\Pi$ that satisfies $\neg\varphi$. 


\def\example{\par\noindent{\bf Example 1.} \ignorespaces}
\def\endexample{}

\begin{example}
\begin{figure}[t]
			\centering
			\subfigure[]{\includegraphics[scale=0.6, height = 4.4cm]{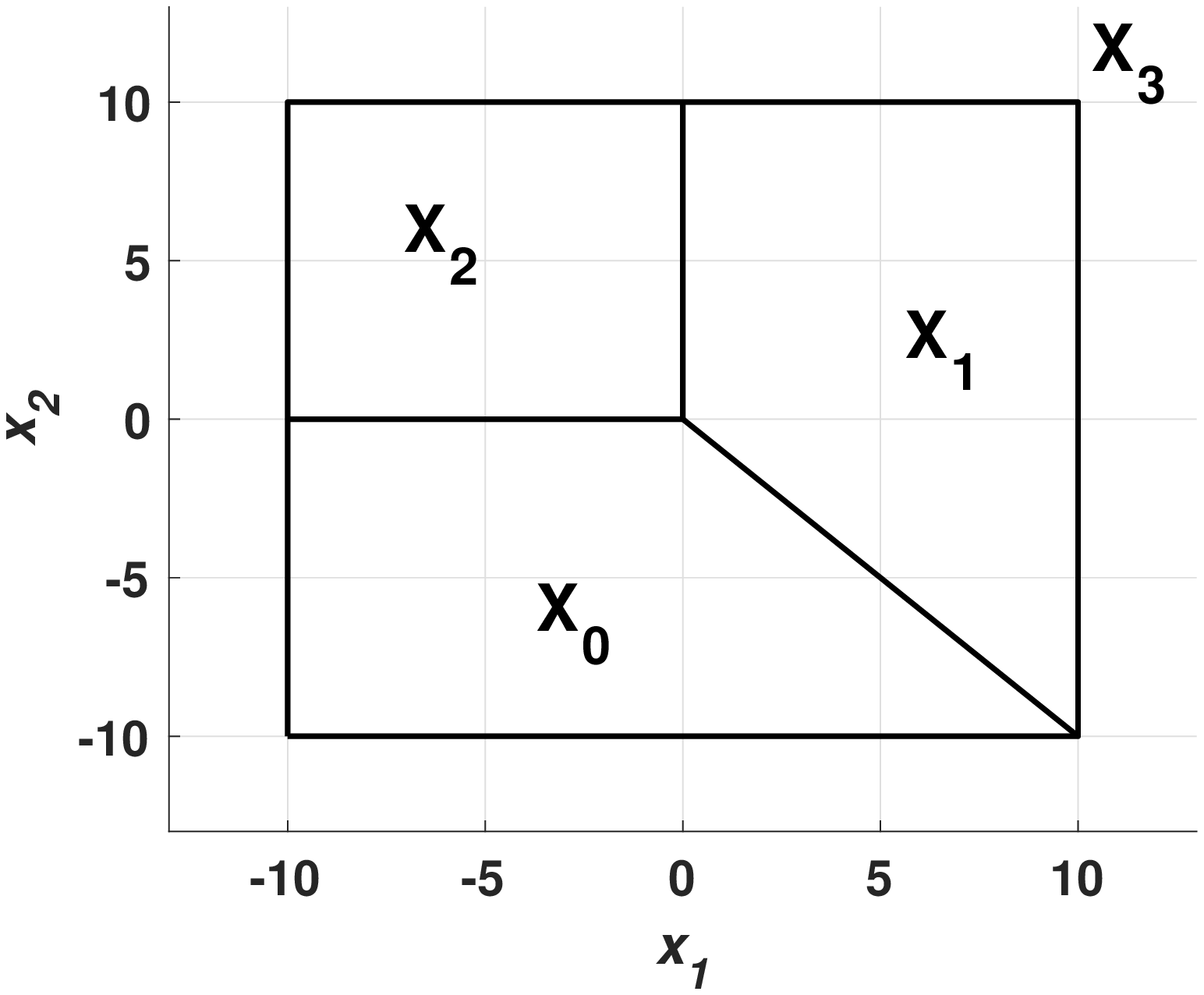}}  \hspace{.5em}
			\subfigure[]{\includegraphics[scale=0.5, height = 4.4cm]{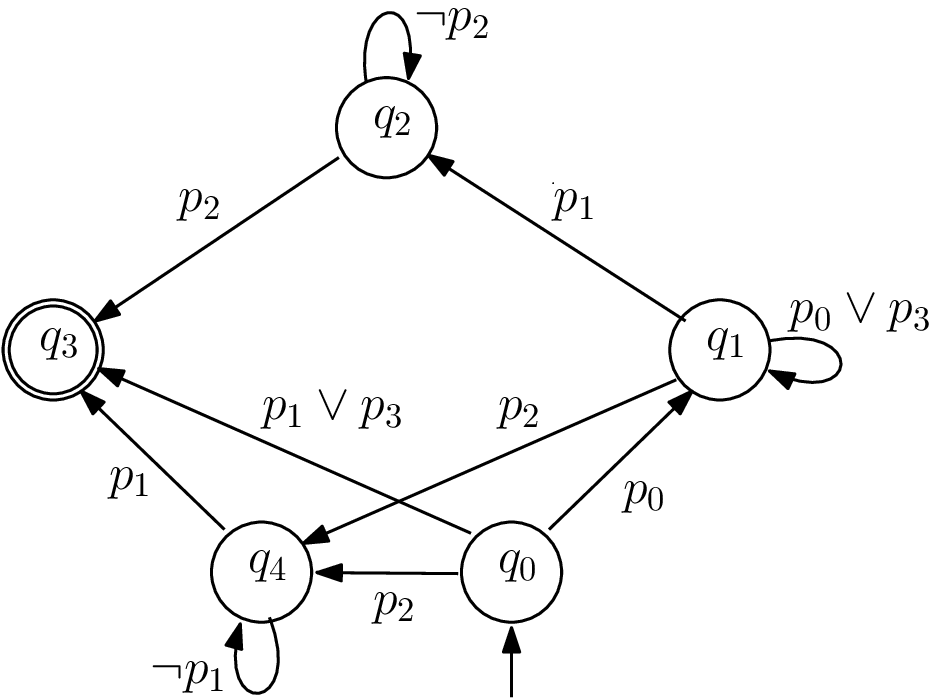}}\\
			\caption{(a) State space and regions of interest for Example~1, (b) DFA $\mathcal{A}_{\neg \varphi}$ that accepts all traces satisfying $\neg \varphi$ where $\varphi$ is given in \eqref{LTL_f}.}
\end{figure}
Consider a two-dimensional stochastic system $S = (X,V_{\textsf w}, w,f)$ with $X = V_{\textsf w} = \mathbb R^2$
and dynamics
\begin{align}
x_1(k+1)&=x_1(k)-0.01 x_2^2(k)+0.1 w_1(k),\nonumber\\
x_2(k+1)&=x_2(k)-0.01 x_1(k) x_2(k)+0.1 w_2(k),
\end{align}
where $w_1(\cdot)$, $w_2(\cdot)$ are independent standard normal random variables.
Let the regions of interest be given as 
\begin{align*}
X_0&=\{(x_1,x_2)\in X \mid x_1\geq-10, \ -10\leq x_2\leq 0, \text{ and }x_1+x_2\leq 0\},\\
X_1&=\{(x_1,x_2)\in X \mid 0\leq x_1\leq 10, \ x_2\leq 10, \text{ and }x_1+x_2\geq 0 \},\\
X_2&=\{(x_1,x_2)\in X \mid -10 \leq x_1\leq 0 \text{ and } 0\leq x_2\leq 10\}, \text{ and }\\
X_3& = X\setminus (X_0\cup X_1\cup X_2).
\end{align*}
The sets $X_0$, $X_1$, $X_2$, and $X_3$ are shown in Figure 1(a).
%
The set of atomic propositions is given by $\Pi=\{p_0,p_1,p_2,p_3\}$, with labeling function $L(x) = p_i$ for any $x\in X_i$, $i\in\{0,1,2,3\}$. We are interested in computing a lower bound on the probability that $Trace_{N}(S)$ of length $N$ satisfies the following specification: 
\begin{itemize}
\item Solution process should start in either $X_0$ or $X_2$. If it starts in $X_0$, it will always stay away from $X_1$ or always stay away from $X_2$. If it starts in $X_2$, it will always stay away from $X_1$ within time horizon $[0,N)\subset\mathbb{N}_0$. 
\end{itemize}
This property can be expressed by the safe LTL$_F$ formula
\begin{equation}
\label{LTL_f}
\varphi=(p_0\wedge(\square \neg p_1 \vee \square\neg p_2))\vee (p_2\wedge\square\neg p_1).
\end{equation}
The DFA corresponding to the negation of the safe LTL$_F$ formula $\varphi$ in \eqref{LTL_f} is shown in Figure 1(b).
\qed
\end{example}
%
%
Next, we provide a systematic approach to solve Problem \ref{prob} by combining automata and barrier certificates introduced in the next section.
We introduce the notion of barrier certificate similar to the one used in \cite{4287147} and show how to use it for solving Problem~\ref{prob} in Sections~\ref{runs}-\ref{BC_computation}.

\section{Barrier Certificate}
\label{bar}
We recall that a function $B:X\rightarrow\mathbb{R}$ is a \textit{supermartingale} for system $S=(X,V_{\textsf w},w,f)$ if
$$\mathbb{E}[B(x(k+1))\mid x(k)]\leq B(x(k)),\quad \forall x(k)\in X,\,k\in\mathbb N_0,$$
where the expectation is with respect to $w(k)$.
This inequality requires that the expected value of $B(x(\cdot))$ does not increases as a function of time.
To provide results for finite time horizon, we instead use a relaxation of supermartingale condition called \textit{c-martingale}.
\begin{definition}
Function $B:X\rightarrow\mathbb{R}$ is a \textit{c-martingale} for system $S=(X,V_{\textsf w},w,f)$ if it satisfies
$$\mathbb{E}[B(x(k+1))\mid x(k)]\leq B(x(k))+c,\quad \forall x(k)\in X,\,k\in\mathbb N_0,$$
with $c\ge 0$ being a non-negative constant.
\end{definition}
We provide the following lemma and use it in the sequel. This lemma is a direct consequence of \cite[Theorem 1]{kushner1965stability} and is also utilized in \cite[Theorem II.1]{steinhardt2012finite}.
\begin{lemma}\label{lemma_1}
Let $B: X\rightarrow \mathbb R_0^{+}$ be a non-negative c-martingale for system $S$. Then for any constant $\lambda>0$ and any initial condition $x_0\in X$, 
\begin{equation}
\mathbb{P}\{\sup_{0\leq k \leq T_d}B(x(k))\geq\lambda\mid x(0)=x_0\}\leq\frac{B(x_0)+cT_d}{\lambda}.
\label{eq_lem1}
\end{equation}
\end{lemma}
%
Next theorem provides inequalities on a barrier certificate that gives an upper bound on reachability probabilities. This theorem is inspired by the result of \cite[Theorem 15]{4287147} that uses supermartingales for reachability analysis of continuous-time systems.

\begin{theorem}
\label{barrier}
Consider a discrete-time stochastic system $S=(X,V_{\textsf w},w,f)$ and sets $X_0, X_1\subseteq X$. Suppose there exist a non-negative function $B:X\rightarrow\mathbb{R}_0^+$ and constants $c\ge 0$ and $\gamma\in[0,1]$ such that
\begin{align}
&B(x)\leq\gamma &\forall x\in X_0,\label{ine_1}\\
&B(x)\geq1 &\forall x\in X_1,\label{ine_2}\\
&B(x) \text{ is c-martingale} &\forall x\in X.\label{ine_3}
\end{align}
Then the probability that the solution process $\textbf{x}_{T_d}$ of $S$ starts from initial state $x(0)\in X_0$ and reaches $X_1$ within time horizon $[0,T_d]\subset \N_0$ is upper bounded by $\gamma+cT_d$.
\end{theorem}
\begin{proof}
Since $B(x(k))$ is non-negative and $c$-martingale, we conclude that \eqref{eq_lem1} in Lemma \ref{lemma_1} holds. Now using \eqref{ine_1} and the fact that $X_1\subseteq \{x\in X \mid B(x) \geq 1 \}$, we have $\mathbb{P}\{x(k)\in X_1 \text{ for some } 0\leq k\leq T_d\mid x(0)=x_0\}$ $\leq\mathbb{P}\{ \sup_{0\leq k\leq T_d} B(x(k))\geq 1\mid x(0)=x_0 \}$ $\leq B(x_0)+cT_d$ $\leq \gamma+cT_d$. This concludes the proof.  \qed
\end{proof}

Theorem~\ref{barrier} enables us to formulate an optimization problem by minimizing the value of $\gamma$ and $c$ in order to find an upper bound for finite-horizon reachability that is as tight as possible. 

In the next section, we discuss how to translate LTL$_F$ verification problem into the computation of a collection of barrier certificates each satisfying inequalities of the form \eqref{ine_1}-\eqref{ine_3}.
Then we show in Section~\ref{BC_computation} how to use Theorem \ref{barrier} to provide a lower bound on the probability of satisfying LTL$_F$ specifications over finite time horizon.


\section{Decomposition into Sequential Reachability}
\label{runs}
Consider a DFA $\mathcal{A}_{\neg \varphi}=(Q,Q_0,\Pi,\delta,F)$ that accepts all finite words of length $n\in[0,N]\subset\N_0$ over $\Pi$ that satisfy $\neg \varphi$.
Self-loops in the DFA play a central role in our decomposition.
Let $Q_s\subseteq Q$ be a set of states of $\mathcal{A}_{\neg \varphi}$ having self-loops, \textit{i.e.}, $Q_s:= \{q\in Q \,|\, \exists p\in \Pi, q\overset{p}{\longrightarrow} q\}$.

\noindent\textbf{Accepting state run of $\mathcal{A}_{\neg \varphi}$.} Sequence $\textbf{q}=(q_0,q_1,\ldots,q_n)\in Q^{n+1}$ is called an accepting state run if $q_0\in Q_0$, $q_n\in F$, and there exist a finite word $\sigma = (\sigma_0,\sigma_1,\ldots,\sigma_{n-1})\in\Pi^n$ such that $q_i \overset{\sigma_i}{\longrightarrow} q_{i+1}$ for all $i\in\{0,1,\ldots, n-1\}$.
We denote the set of such finite words by $\sigma(\textbf{q})\subseteq \Pi^n$ and the set of accepting runs by $\mathcal R$. We also indicate the length of $\textbf{q}\in Q^{n+1}$ by $|\textbf{q}|$, which is $n+1$.

Let $\mathcal{R}_{\leq N+1}$ be the set of all finite accepting state runs of lengths less than or equal to $N+1$ excluding self-loops,
\begin{equation}
\label{eq:runs}
\mathcal{R}_{\leq N+1} := \{\textbf{q} = (q_0,q_1,\ldots,q_n)\in\mathcal R \,|\, n\le N,\, q_i\neq q_{i+1},\, \forall i<n\}.
\end{equation}
Computation of $\mathcal{R}_{\leq N+1}$ can be done efficiently using algorithms in graph theory by viewing $\mathcal{A}_{\neg \varphi}$ as a directed graph.
%
%
%
%
Consider $\mathcal{G}=(\mathcal{V},\mathcal{E})$ as a directed graph with vertices $\mathcal{V}=Q$ and edges $\mathcal{E}\subseteq\mathcal{V}\times\mathcal{V}$ such that $(q,q')\in\mathcal{E}$ if and only if $q'\neq q$ and there exist $p\in\Pi$ such that $q\overset{p}{\longrightarrow} q'$. From the construction of the graph, it is obvious that the finite path in the graph of length $n+1$ starting from vertices $q_0\in Q_0$ and ending at $q_F\in F$ is an accepting state run $\textbf{q}$ of $\mathcal{A}_{\neg \varphi}$ without any self-loop thus belongs to $\mathcal{R}_{\leq N+1}$.
Then one can easily compute $\mathcal{R}_{\leq N+1}$ using variants of depth first search algorithm \cite{russell2003artificial}.

Decomposition into sequential reachability is performed as follows.
For any $\textbf{q} = (q_0,q_1,\ldots,q_n)\in\mathcal{R}_{\leq N+1}$, we define $\mathcal{P}(\textbf{q})$ as a set of all state runs of length $3$ augmented with a horizon,
\begin{equation}
\label{eq:reachability}
\mathcal{P}(\textbf{q}):=\{\left(q_i,q_{i+1},q_{i+2},T(\textbf{q},q_{i+1})\right)\mid 0\leq i\leq n-2\},
\end{equation}
where the horizon is defined as $T(\textbf{q},q_{i+1}) = N+2-|\textbf{q}|$ for $q_{i+1}\in Q_s$ and $1$ otherwise.
\begin{remark}
Note that $\mathcal{P}(\textbf{q})=\emptyset$ for $|\textbf{q}|=2$. In fact, any accepting state run of length $2$ specifies a subset of the state space such that the system satisfies $\neg\varphi$ whenever it starts from that subset. This gives trivial zero probability for satisfying the specification, thus neglected in the sequel. 
\end{remark}


The computation of sets $\mathcal{P}(\textbf{q})$, $\textbf{q}\in \mathcal R_{\le N+1}$, is illustrated in Algorithm \ref{algo1} and demonstrated below for our demo example.

\begin{algorithm}[t]
	\caption{Computation of sets $\mathcal{P}(\textbf{q})$, $\textbf{q}\in \mathcal R_{\le N+1}$}
	\label{algo1}
	\begin{algorithmic}[1]
		\Require{$ \mathcal{G} $, $ Q_s $, $N$}
		\State Compute set $\mathcal{R}_{\leq N+1}$ by depth first search on $ \mathcal{G} $
		\ForAll{$ \textbf{q}\in \mathcal{R}_{\leq N+1} $ and $|\textbf{q}|\geq3$}
		\For {$i=1$ to $|\textbf{q}|-3$}
		\State $\mathcal{P}_1(\textbf{q})\leftarrow \{(q_i,q_{i+1},q_{i+2})\}$
		\If {$q_{i+1}\in Q_s$}
		\State $\mathcal{P}(\textbf{q})\leftarrow \{(q_i,q_{i+1},q_{i+2},N+2-|\textbf{q}|)\}$
		\Else
		\State $\mathcal{P}(\textbf{q})\leftarrow \{(q_i,q_{i+1},q_{i+2},1)\}$
		\EndIf
		\EndFor
		\EndFor
		\Return $\mathcal{P}(\textbf{q})$
	\end{algorithmic}
\end{algorithm}

\def\example{\par\noindent{\bf Example 1.} \ignorespaces}
\def\endexaple{}

\begin{example}(continued)
For safe LTL$_F$ formula $\varphi$ given in \eqref{LTL_f}, Figure 1(b) shows a DFA $\mathcal{A}_{\neg\varphi}$ that accepts all words that satisfy $\neg\varphi$. From Figure 1(b), we get $Q_0=\{q_0\}$ and $F=\{q_3\}$. We consider traces of maximum length $N=5$.
The set of accepting state runs of lengths at most $N+1$ without self-loops is
 \begin{equation*}
 \mathcal{R}_{\leq 6}=\{(q_0,q_4,q_3),(q_0,q_1,q_2,q_3),(q_0,q_1,q_4,q_3),(q_0,q_3)\}.
 \end{equation*}
The set of states with self-loops is $Q_s = \{q_1,q_2,q_4\}$. Then the sets $\mathcal{P}(\textbf{q})$ for $\textbf{q}\in\mathcal{R}_{\leq 6}$ are as follows:
\begin{align*}
&\mathcal{P}(q_0,q_3)=\emptyset,\quad \mathcal{P}(q_0,q_4,q_3)=\{(q_0,q_4,q_3,4)\},\\
&\mathcal{P}(q_0,q_1,q_2,q_3)=\{(q_0,q_1,q_2,3),(q_1,q_2,q_3,3)\},\\
&\mathcal{P}(q_0,q_1,q_4,q_3)=\{(q_0,q_1,q_4,3),(q_1,q_4,q_3,3)\}.
\end{align*} 
%
For every $\textbf{q}\in\mathcal{R}_{\leq 6}$, the corresponding finite words $\sigma(\textbf{q})$ are listed as follows:
\begin{align*}
& \sigma(q_0,q_3)=\{p_1\vee p_3\},\quad \sigma(q_0,q_4,q_3)=\{(p_2,p_1)\},\\
& \sigma(q_0,q_1,q_2,q_3)=\{(p_0,p_1,p_2)\},\quad \sigma(q_0,q_1,q_4,q_3)=\{(p_0,p_2,p_1)\}.
\end{align*}
\end{example}\qed

\section{Computation of Probabilities Using Barrier Certificates}
\label{BC_computation}
Having the set of state runs of length $3$ augmented with horizon, in this section, we provide a systematic approach to compute a lower bound on the probability that the solution process of $S$ satisfies $\varphi$. Given DFA $\mathcal{A}_{\neg\varphi}$, our approach relies on performing a reachability computation over each element of $\mathcal P(\textbf{q})$, $\textbf{q}\in\mathcal R_{\le N+1}$, where reachability probability is upper bounded using barrier certificates. 

Next theorem provides an upper bound on the probability that the solution process of the system satisfies the specification $\neg\varphi$.
\begin{theorem}\label{upper_bound}
For a given safe LTL$_F$ specification $\varphi$, let $\mathcal{A}_{\neg\varphi}$ be a DFA corresponding to its negation, $\mathcal R_{\le N+1}$ be the set of accepting state runs of length at most $N+1$ as defined in \eqref{eq:runs}, and $\mathcal P$ be the set of runs of length $3$ augmented with horizon as defined in \eqref{eq:reachability}.  Then the probability that the system satisfies $\neg\varphi$ within time horizon $[0,N]\subseteq\mathbb{N}_0$ is upper bounded by
\begin{equation}
\label{eq:bound}
\mathbb{P}\{Trace_{N}(S)\models \neg \varphi \} \leq\hspace{-0.5em} \sum_{\textbf{q} \in\mathcal R_{\le N+1}}\hspace{-0.5em}\prod\left\{(\gamma_\nu+c_\nu T)\,|\, \nu = (q,q',q'',T)\in \mathcal{P}(\textbf{q})\right\},
\end{equation}
where $\gamma_\nu+c_\nu T$ is the upper bound on the probability of the trajectories of $S$ starting from $X_0 := L^{-1}(\sigma(q,q'))$ and reaching $ X_1:=L^{-1}(\sigma(q',q''))$ within time horizon $[0,T]\subseteq\mathbb{N}_0$ computed via Theorem~\ref{barrier}.
\end{theorem}
\begin{proof}
Consider an accepting run $\textbf{q}\in \mathcal{R}_{\leq N+1}$ and set $\mathcal{P}(\textbf{q})$ as defined in \eqref{eq:reachability}. For an element $\nu=(q,q',q'',T)\in \mathcal{P}(\textbf{q})$, the upper bound on the probability of trajectories of $S$ stating from $L^{-1}(\sigma(q,q'))$ and reaching $L^{-1}(\sigma(q',q''))$ within time horizon $T$ is given by $\gamma_\nu+c_\nu T$. This follows from Theorem \ref{barrier}. Now the upper bound on the probability of the trace of the solution process reaching accepting state following trace corresponding to $\textbf{q}$ is given by the product of the probability bounds corresponding to all elements $\nu=(q,q',q'',T)\in \mathcal{P}(\textbf{q})$ and is given by
  \begin{equation}
\label{eq:bound_1}
\mathbb{P}\{ \sigma_{\textbf{x}_{N}}(\textbf{q})\models \neg \varphi \} \leq \prod \left\{(\gamma_\nu+c_\nu T)\,|\, \nu = (q,q',q'',T)\in \mathcal{P}(\textbf{q})\right\}.
\end{equation}
Note that, the way we computed time horizon $T_d$, we always get the upper bound for the probabilities for all possible combinations of self-loops for accepting state runs of length less than or equal to $N+1$.
The upper bound on the probability that the solution processes of system $S$ violate $\varphi$ can be computed by summing the probability bounds for all possible accepting runs as computed in \eqref{eq:bound_1} and is given by
\begin{equation}
\label{eq:bound_2}
\mathbb{P}\{Trace_{N}(S)\models \neg \varphi \} \leq \sum_{\textbf{q} \in\mathcal R_{\le N+1}}\prod\left\{(\gamma_\nu+c_\nu T)\,|\, \nu = (q,q',q'',T)\in \mathcal{P}(\textbf{q})\right\}.\nonumber
\end{equation} \qed
\end{proof}

Theorem~\ref{upper_bound} enables us to decompose the computation into a collection of sequential reachability, compute bounds on the reachability probabilities using Theorem~\ref{barrier}, and then combine the bounds in a sum-product expression.
\begin{remark}
In case we are unable to find barrier certificates for some of the elements $\nu \in \mathcal P(\textbf q)$ in \eqref{eq:bound}, we replace the related term $(\gamma_\nu+c_\nu T)$ by the pessimistic bound $1$. In order to get a non-trivial bound in \eqref{eq:bound}, at least one barrier certificate must be found for each $\textbf{q} \in\mathcal R_{\le N+1}$.
\end{remark}
\begin{corollary}
\label{lower_bound}
Given the result of Theorem \ref{upper_bound}, the probability that the trajectories of $S$ of length $N$ satisfies safe LTL$_F$ specification $\varphi$ is lower-bounded by
\begin{equation}
\mathbb{P}\{Trace_{N}(S)\models \varphi \}\geq1-\mathbb{P}\{Trace_{N}(S)\models \neg \varphi \}.\nonumber
\end{equation}
\end{corollary}

\subsection{Computation of barrier certificate}
Proving existence of a barrier certificate, finding one, or showing that a given function is in fact a barrier certificate are in general hard problems. But if we restrict the class of systems and labeling functions, we can construct computationally efficient techniques for searching barrier certificates of specific forms. One technique is to use sum-of-squares (SOS) optimization \cite{parrilo2003semidefinite}, which relies on the fact that a polynomial is non-negative if it can be written as sum of squares of different polynomials. Therefore, we raise the following assumption.
\begin{assumption}
\label{ass:BC}
System $S$ has state set $X\subseteq \mathbb R^n$ and its vector field $f:X\times V_{\textsf{w}}\rightarrow X$ is a polynomial function of state $x$ for any $w\in V_{\textsf{w}}$.
Partition sets $X_i = L^{-1}(p_i)$, $i\in\{0,1,2,\ldots, M\}$, are bounded semi-algebraic sets, \textit{i.e.}, they can be represented by polynomial equalities and inequalities.
\end{assumption}


Under Assumption~\ref{ass:BC}, we can formulate \eqref{ine_1}-\eqref{ine_3} as an SOS optimization problem to search for a polynomial-type barrier certificate $B(\cdot)$ and the tightest upper bound $(\gamma+c T_d)$. The following lemma provides a set of sufficient conditions for the existence of such a barrier certificate required in Theorem \ref{barrier}, which can be solved as an SOS optimization.
\begin{lemma}
\label{sos}
Suppose Assumption~\ref{ass:BC} holds and sets $X_0,X_1,X$ can be defined by vectors of polynomial inequalities
$X_0=\{x\in\R^n\mid g_0(x)\geq0\}$, $X_1=\{x\in\R^n\mid g_1(x)\geq0\}$, and $X=\{x\in\R^n\mid g(x)\geq0\}$, where the inequalities are defined element-wise.
Suppose there exists a sum-of-squares polynomial B(x), constants $\gamma\in [0,1]$ and $c\ge 0$, and vectors of sum-of-squares polynomials $\lambda_0(x)$, $\lambda_1(x)$, and $\lambda(x)$ of appropriate size such that following expressions are sum-of-squares polynomials
\begin{align}
-B(x)-\lambda_0^T(x) g_0(x)+\gamma\label{eq:sos1}\\
B(x)-\lambda_1^T(x) g_1(x)-1\\
-\mathbb{E}[B(f(x,w))|x]+B(x)-\lambda^T(x) g(x)+c.\label{eq:sos2}
\end{align}
Then  $B(x)$ satisfies conditions \eqref{ine_1}-\eqref{ine_3}.
\end{lemma}
 \begin{proof}
 The proof is similar to that of Lemma 7 in \cite{7364197} and is omitted due to lack of space.  \qed
 \end{proof}
 \begin{remark}
 Assumption~\ref{ass:BC} is essential for applying the results of Lemma~\ref{sos} to \emph{any} LTL$_F$ specification. For a given specification, we can relax this assumption and allow some of the partition sets $X_i$ to be unbounded. For this, we require that the labels corresponding to unbounded partition sets should only appear either on self-loops or on accepting runs of length less than 3. For instance, Example~1 has an unbounded partition set $X_3$ and its corresponding label $p_3$ satisfies this requirement (see Figure 1), thus the results are still applicable for verifying the specification.
 \end{remark}
 
 \subsection{Computational complexity}
 Based on Lemma \ref{sos}, a polynomial barrier certificate $B(\cdot)$ satisfying $\eqref{ine_1}$-$\eqref{ine_3}$ and minimizing constants $\gamma$ and $c$  can be automatically computed using SOSTOOLS \cite{1184594} in conjunction with a semidefinite programming solver such as SeDuMi \cite{sturm1999using}. We refer the interested reader to \cite{steinhardt2012finite} and  \cite{4287147} for more discussions. Note that the value of the upper bound of violating the property depends highly on the selection of degree of polynomials in Lemma \ref{sos}.
 
 From the construction of directed graph $\mathcal{G}=(\mathcal{V},\mathcal{E})$, explained in Section \ref{runs}, the number of triplets and hence the number of barrier certificates needed to be computed are bounded by $|\mathcal{V}|^3 = |Q|^3$, where $|\mathcal{V}|$ is the number of vertices in $\mathcal{G}$.
 Further, it is known \cite{baier2008principles} that $|Q|$ is at most $|\neg\varphi |2^{|\neg\varphi |}$, where $|\neg\varphi |$ is the length of formula $\neg\varphi$ in terms of number of operations, but in practice, it is much smaller than this bound \cite{klein2006experiments}. 
  
 Computational complexity of finding polynomials $B,\lambda_0,\lambda_1,\lambda$ in Lemma~\ref{sos} depends on both the degree of polynomials appearing in \eqref{eq:sos1}-\eqref{eq:sos2} and the number of variables. It is shown that for fixed degrees the required computations grow polynomially with respect to the dimension \cite{7364197}. Hence we expect that this technique is more scalable in comparison with the discretization-based approaches especially for large-scale systems.
 

\section{Case Studies}
\label{examples}
In this section, we demonstrate the effectiveness of the proposed results on several case studies. We first showcase the results on the running example, which has nonlinear dynamics with additive noise. We then apply the technique to a ten-dimensional linear system with additive noise to show its scalability. The third case study is a three-dimensional nonlinear system with multiplicative noise.
  
\subsection{Running example}
\label{running}
To compute an upper bound on reachability probabilities corresponding to each element of $\mathcal P(\textbf{q})$ in Theorem~\ref{upper_bound}, we use Lemma \ref{sos} to formulate it as a SOS optimization problem to minimize values of $\gamma$ and $c$ using bisection method.
The optimization problem is solved using SOSTOOLS and SeDuMi, to obtain upper bounds in Theorem~\ref{upper_bound}.
The computed upper bounds on probabilities corresponding to the elements of $\mathcal P(\cdot)$, $(q_0,q_4,q_3,4)$, $(q_0,q_1,q_2,3)$, $(q_1,q_2,q_3,3)$, $(q_0,q_1,q_4,3)$, and $(q_1,q_4,q_3,3)$ are respectively $0.00586$, $0.00232$, $0.00449$, $0.00391$, and $0.00488$.
Using Theorem \ref{upper_bound}, we get
\begin{equation}
\mathbb{P}\{Trace_{N}(S)\models \neg \varphi \}\le 0.00586+0.00232\times 0.00449+0.00391\times 0.00488= 0.00589.\nonumber
\end{equation}
Thus, a lower bound on the probability that trajectories of $S$ satisfy safe LTL$_F$ property \eqref{LTL_f} over time horizon $N=5$ is given by $0.99411$.
The optimization finds polynomials of degree $5$ for $B$, $\lambda$, $\lambda_0$, and $\lambda_1$. Hence $4$ barrier certificates are computed each with $245$ optimization coefficients, which takes $29$ minutes in total.
For the sake of comparison, we provide a probabilistic guarantee from Monte-Carlo method using $50000$ realizations, which results in the interval $\mathbb{P}\{Trace_{N}(S)\models \varphi \}\in [0.9959, \text{ }0.9979]$ with confidence $1-10^{-4}$.


 

\subsection{Thermal model of a ten-room building}
\label{thermal}
\begin{figure}[t]
			\centering
			\subfigure[]{\includegraphics[scale=0.3, height = 2.5cm]{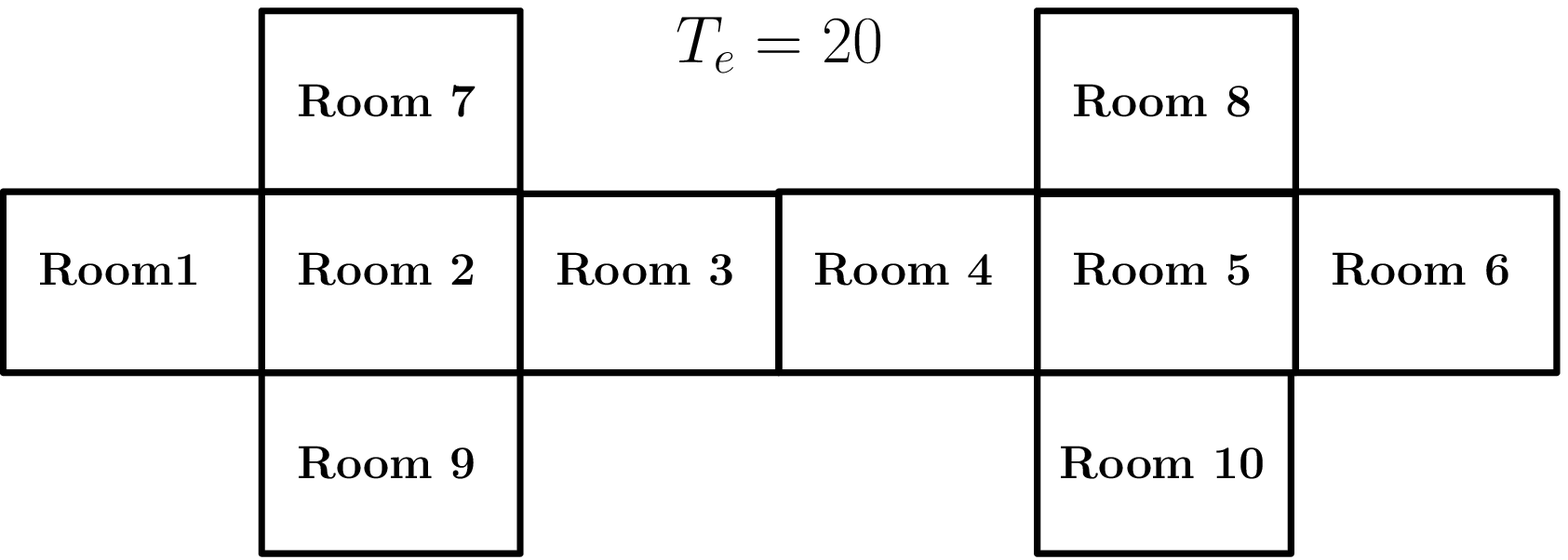}}  \hspace{1em}
			\subfigure[]{\includegraphics[scale=0.48, height = 2.8cm]{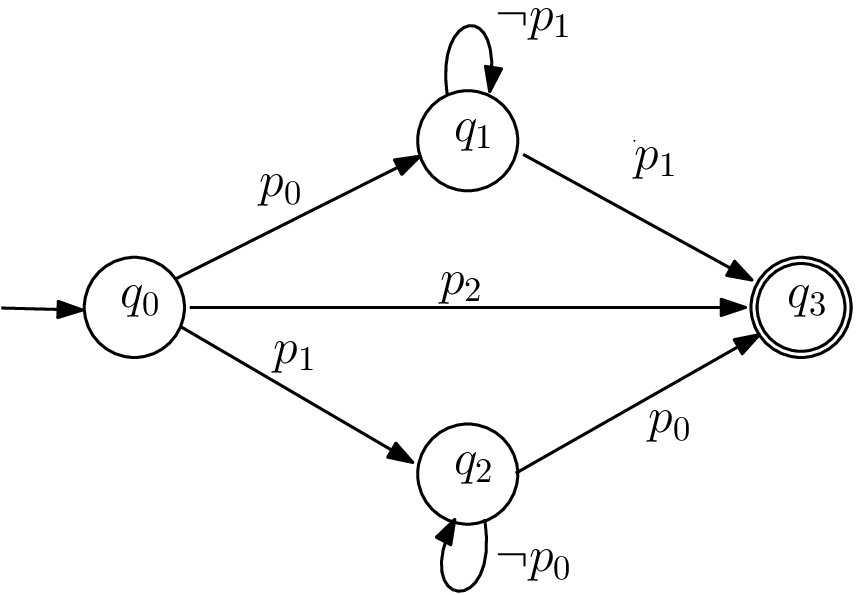}}\\
			\caption{(a) A schematic of ten-room building, (b) DFA $\mathcal{A}_{\neg \varphi}$ that accepts all traces satisfying $\neg \varphi$ where $\varphi$ is given in \eqref{ltl_f_2}.}
\end{figure}
Consider temperature evolution in a ten-room building shown schematically in Figure 2(a). We use this model to demonstrate the effectiveness of the results on large-dimensional state spaces. This model is adapted from \cite{jagtap2017quest} by discretizing it with sampling time $\tau_s=5$ minutes and without including heaters.
The dynamics of $S$ are given as follows:
\begin{align*}
	 x_1  (k+1) &= (1-\tau_s(\alpha+\alpha_{e1})) x_1(k)+\tau_s\alpha x_2(k)+\tau_s\alpha_{e1}T_e+0.5 w_1(k),\\
	 x_2  (k+1)  &= (1-\tau_s(4 \alpha+\alpha_{e2})) x_2(k)+\tau_s\alpha(x_1(k)+x_3(k)+x_7(k)+x_9(k))  \\& \ \ \ + \tau_s\alpha_{e2}T_e+ 0.5 w_2(k),\\
	 x_3  (k+1)  &= (1-\tau_s(2 \alpha+\alpha_{e1})) x_3(k)+\tau_s\alpha(x_2(k)+x_4(k))  \hspace{-.2em}+ \hspace{-.2em}\tau_s\alpha_{e1}T_e\hspace{-.2em}+\hspace{-.2em} 0.5 w_3(k),\\
	 x_4 (k+1)  &= (1-\tau_s(2 \alpha+\alpha_{e1})) x_4(k)+\tau_s\alpha(x_3(k)+x_5(k))  \hspace{-.2em}+ \hspace{-.2em}\tau_s\alpha_{e1}T_e\hspace{-.2em}+\hspace{-.2em} 0.5 w_4(k),\\
	 x_5  (k+1)  &= (1-\tau_s(4 \alpha+\alpha_{e2})) x_5(k)+\tau_s\alpha(x_4(k)+x_6(k)+x_8(k)+x_{10}(k))  \\& \ \ \ + \tau_s\alpha_{e2}T_e+ 0.5 w_5(k),\\
	x_6  (k+1)  &= (1-\tau_s(\alpha+\alpha_{e1})) x_6(k)+\tau_s\alpha x_5(k) + \tau_s\alpha_{e1}T_e+ 0.5 w_6(k),\\
	 x_7  (k+1)  &= (1-\tau_s(\alpha+\alpha_{e1})) x_7(k)+\tau_s\alpha x_2(k) + \tau_s\alpha_{e1}T_e+ 0.5 w_7(k),\\
	 x_8  (k+1)  &= (1-\tau_s(\alpha+\alpha_{e1})) x_8(k)+\tau_s\alpha x_5(k) + \tau_s\alpha_{e1}T_e+ 0.5 w_8(k),\\
	x_9  (k+1)  &= (1-\tau_s(\alpha+\alpha_{e1})) x_9(k)+\tau_s\alpha x_2(k) + \tau_s\alpha_{e1}T_e+ 0.5 w_9(k),\\
	 x_{10}  (k+1)  &= (1-\tau_s(\alpha+\alpha_{e1})) x_{10}(k)+\tau_s\alpha x_5(k) + \tau_s\alpha_{e1}T_e+ 0.5 w_{10}(k),
\end{align*} 
where
$x_i$, $i \in \{ 1, 2, \ldots,10\}$, denotes the temperature in each room,
$T_e=20^\circ$C is the ambient temperature, and $\alpha=5\times10^{-2}$, $\alpha_{e1}=5\times 10^{-3}$, and $\alpha_{e2}=8\times 10^{-3}$ are heat exchange coefficients.

Noise terms $w_i(k)$, $i\in\{1,2,\ldots,10\}$, are independent standard normal random variables.
The state space of the system is $X=\R^{10}$. We consider regions of interest $X_0 = [18,19.75]^{10}$, $X_1 = [20.25,22]^{10}$, $X_2 = X\setminus (X_0\cup X_1)$. 
The set of atomic propositions is given by $\Pi=\{p_0,p_1,p_2\}$ with labeling function $L(x_i) = p_i$ for all $x_i\in X_i$, $i\in\{0,1,2\}$.
The objective is to compute a lower bound on the probability that the solution process of length $N=50$ satisfies the safe LTL$_F$ formula
\begin{equation}\label{ltl_f_2}
\varphi= (p_0\wedge\square\neg p_1)\vee (p_1\wedge\square\neg p_0).
\end{equation}
The DFA $\mathcal{A}_{\neg \varphi}$ corresponding to $\neg \varphi$ is shown in Figure 2(b). We use Algorithm~\ref{algo1} to get
$\mathcal R_{\le 11} = \{(q_0,q_3),(q_0,q_1,q_3),(q_0,q_2,q_3)\}$,
$\mathcal{P}(q_0,q_1,q_3)=\{q_0,q_1,q_3,9\}$, and $\mathcal{P}(q_0,q_2,q_3)=\{q_0,q_2,q_3,9\}$. As described in Section \ref{BC_computation},
we compute two barrier certificates and SOS polynomials satisfying inequalities of Lemma \ref{sos}.
The lower bound $\mathbb{P}\{Trace_{N}(S)\models \varphi \}\ge 0.9820$ is obtained using SOSTOOLS and SeDuMi for initial states starting from $X_0\cup X_1$. 
The optimization finds $B$, $\lambda$, $\lambda_0$, and $\lambda_1$ as quadratic polynomials. Hence two barrier certificates are computed each with $255$ optimization coefficients, which takes $18$ minutes in total.
For the sake of comparison, we provide a probabilistic guarantee from Monte-Carlo method using $50000$ realizations, which results in the interval $\mathbb{P}\{Trace_{N}(S)\models \varphi \}\in [0.9984, \text{ }0.9997]$ with confidence $1-10^{-5}$.

\begin{figure}[t]
			\centering
			\includegraphics[scale=0.55]{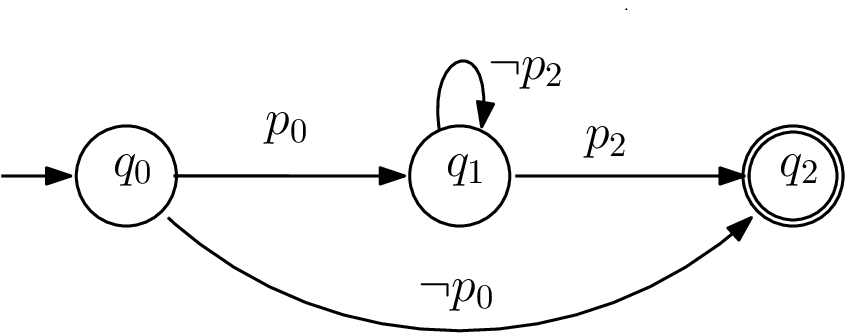}
			\caption{DFA $\mathcal{A}_{\neg \varphi}$ that accepts all traces satisfying $\neg \varphi$ where $\varphi=p_0\wedge\square\neg p_2$.}
			\label{DFA_3}
\end{figure}
\subsection{Lorenz model of a thermal convection loop}
Our third case study is the Lorenz model of a thermal convection loop as used in \cite{7799318} with multiplicative noise. The nonlinear dynamics of $S$ is given as 
\begin{align}
x_1(k+1)&=(1-aT)x_1(k)+aTx_2(k)+0.025x_1(k)w_1(k),\nonumber\\
x_2(k+1)&=(1-T)x_2(k)-Tx_2(k)x_3(k)+0.025x_2(k)w_2(k),\nonumber\\
x_3(k+1)&=(1+bT)x_3(k)+Tx_1(k)x_2(k)+0.025x_3(k)w_3(k),\label{eq:Lorenz}
\end{align}
where $a=10$, $b=8/3$, and $T=0.01$. Noise terms $w_1(k)$, $w_2(k)$, and $w_3(k)$ are independent standard normal random variables.
We refer the interested readers to \cite{Vincent1991} for a detailed treatment of the model. The state space of the system is $X=\R^3$. We define regions of interest as $X_0 =[-10,10]^2\times [2,10]$, $X_1 = [-10,10]^2\times [-2,2]$, $X_2 = [-10,10]^2\times [-10,-2]$, and $X_3 = X \setminus (X_0\cup X_1\cup X_2)$.

The set of atomic propositions is given by $\Pi=\{p_0,p_1,p_2,p_3\}$ with labeling function $L(x_i) = p_i$ for all $x_i\in X_i$, $i\in\{0,1,2,3\}$.
We consider safe LTL$_F$ property $\varphi=p_0\wedge\square\neg p_2$ and time horizon $N=10$.
The DFA $\mathcal{A}_{\neg \varphi}$ corresponding to the negation of $\varphi$ is shown in Figure \ref{DFA_3}. 
One can readily see that $\mathcal R_{\le 11} = \{(q_0,q_1,q_2)\}$ with $\mathcal P(q_0,q_1,q_2) = (q_0,q_1,q_2,9)$. Thus, we need to compute only one barrier certificate. 
We use inequalities of Lemma \ref{sos} and find a barrier certificate that gives the lower bound
$\mathbb{P}\{Trace_{N}(S)\models \varphi \}\ge 0.9859$.
The optimization finds $B$, $\lambda$, $\lambda_0$, and $\lambda_1$ as polynomials of degree $4$. Hence only one barrier certificate is computed with $53$ optimization coefficients, which takes $3$ minutes.
For the sake of comparison, Monte-Carlo method results in the interval $[0.9912, \text{ }0.9972]$ for the true probability with confidence $1-10^{-4}$ using $10000$ realizations.


Remark that current implementations of discretization-based approaches (e.g., \cite{Faust}) are not directly applicable to the model in Subsection~\ref{running} and to the model~\eqref{eq:Lorenz} due to the multiplicative noise in the latter and unbounded state space of the former. Application of these techniques to the model in Subsection~\ref{thermal} will also be computationally much more expensive than our approach due to its exponential complexity as a function of state space dimension.


\section{Conclusions}
%
In this paper, we proposed a discretization-free approach for formal verification of discrete-time stochastic systems. The approach computes lower bounds on the probability of satisfying a specification encoded as safe LTL over finite traces. It is based on computation of barrier certificates and uses sum-of-squares optimization to find such bounds.
From the implementation perspective, we plan to generalize our code and make it publicly available so that it can be applied to systems and specifications defined by users.

\bibliographystyle{splncs03}

\bibliography{IEEEtran1}

\end{document}